\RequirePackage{xcolor}
\documentclass[9pt,shortpaper,twoside,web]{IEEEtran}
\usepackage{generic}
\usepackage{cite}
\usepackage{amsmath,amssymb,amsfonts}
\usepackage{bbm}
\usepackage{algorithmic}
\usepackage{graphicx}
\usepackage{textcomp}

\usepackage{tikz-cd}

\def\BibTeX{{\rm B\kern-.05em{\sc i\kern-.025em b}\kern-.08em
    T\kern-.1667em\lower.7ex\hbox{E}\kern-.125emX}}
\newtheorem{defn}{Definition}

\newtheorem{lemma}{Lemma}
\newtheorem{theorem}{Theorem}
\newtheorem{corollary}{Corollary}
\newtheorem{remark}{Remark}

\newcommand{\eg}{\emph{e.g.}}
\newcommand{\ie}{\emph{i.e.}}
\newcommand{\cf}{\emph{cf.}}
\newcommand{\concat}[3]{#1 +\hspace{-3pt}^{#3} #2}
\newcommand{\nerode}[1]{\stackrel{#1}{\sim}}

\newcommand{\newstuff}[1]{{#1}}

\begin{document}
\title{On state-space representations of general discrete-time dynamical systems}
\author{Cristian R. Rojas, \IEEEmembership{Member, IEEE}, and Pawe\l{} Wachel, \IEEEmembership{Member, IEEE}
\thanks{This work was partially supported by the Swedish Research Council under contract number 2016-06079 (NewLEADS) and by the Digital Futures project EXTREMUM.}
\thanks{C. R. Rojas is with the School of Electrical Engineering and Computer Science, KTH Royal Institute of Technology, 100 44 Stockholm, Sweden (e-mail: cristian.rojas@ee.kth.se). }
\thanks{P. Wachel is with the Department of Control Systems and Mechatronics, Wroc\l{}aw University of Science and Technology, Wroc\l{}aw, Poland (e-mail: pawel.wachel@pwr.edu.pl).}}

\maketitle

\begin{abstract}
In this paper we establish that every (deterministic) non-autonomous, discrete-time, causal, time invariant system has a state-space representation\newstuff{, and discuss its minimality}.
\end{abstract}


\section{Introduction} \label{sec:introduction}
The concept of dynamical systems, or systems which can evolve over time, is crucial in many fields of science and engineering, including statistics, economics, control theory and computer science (under the name of \emph{automata}~\cite{Hopcroft-Motwani-Ullman-06}). Many formulations of a dynamical system include the notion of a \emph{state}, as a vector quantity whose values capture the whole ``history'' of the system up to a given time instant, \eg, \cite{Arnold-92,Bellman-57,Bellman-61,Fuller-60,Katok-Hasselblatt-95,LaSalle-76,Luenberger-79,Sontag-98,Mesarovic-Takahara-75}. A formulation of a system that describes how its state evolves over time is informally called a \emph{state-space representation} of the system.

Within the field of control theory, state-space descriptions for non-autonomous systems (\ie, with an external input signal) became popular thanks to the work of R. E. Kalman~\cite{Kalman:60a}, who developed a general theory of analysis and control of linear systems in state-space form. According to Kalman, the state vector of a general (non-autonomous) dynamical system corresponds to the smallest set of values, specified at a given time $t=t_0$, which, together with the value of the inputs for every time $t \geq t_0$, are enough to predict the behaviour of the system for any time $t \geq t_0$; see also \cite{Zadeh-Desoer-63}. This notion plays a fundamental role in control, as it allows, \eg, to decompose the design of a controller for a system in terms of a state estimator and a (static) state feedback block.

A general non-autonomous dynamical system can be defined in a more general way as an \emph{arbitrary mapping} between input and output signals~\cite{Oppenheim-Schafer-10}. A natural question that arises is whether every such system can be described in state-space form, and if so, how such a representation can be found; this is called the \emph{state-space realization problem}. Within the class of linear systems, this problem was first solved by Ho and Kalman in \cite{Ho-Kalman-66}, who provided a constructive algorithm to obtain such a representation. Ho-Kalman's algorithm provided the foundation for the area of subspace identification algorithms~\cite{Kung-78, Katayama-05}, which allow one to estimate models of linear systems directly in state-space form.

Further investigations on the state-space realization problem for linear systems led to an abstract input-output characterization of the notion of state, based on the concept of \emph{Nerode equivalence}~\cite{Nerode-58}, according to which the state is related to the set of input sequences which, being equal for all times $t \geq t_0$, yield the same output sequence for all $t \geq t_0$~\cite{Kailath-80}. See \cite{Kalman-Falb-Arbib-69} for a presentation of this concept at the intersection of automata theory and mathematical system theory.

While the realization problem can be considered to be solved for linear systems, to the best of our knowledge it has not been addressed for general (possibly nonlinear) non-autonomous dynamical systems. The goal of this note is to fill this gap, by studying the general discrete-time state-space realization problem. Specifically, our contribution is to show that every non-autonomous discrete-time system which is time invariant and causal can be described in state-space form\newstuff{, by using the notion of Nerode equivalence, and then to discuss the minimality of such Nerode representation}.

\newstuff{It is worth mentioning that the realization problem has also been intensely studied within \emph{behavioural system theory}~\cite{Polderman-Willems-98}, in particular for linear systems~\cite{Verriest-20}, where a system is defined as a set of \emph{behaviours}, \ie, a subset of the class of all trajectories that its external signals can take\footnote{Behavioural system theory does not make a distinction between input and output signals.}. Since the definitions in this field differ from those used in input-output system theory, we will not pursue this line of research in this paper.}

The paper is structured as follows. Section~\ref{sec:notation} introduces the main notation and definitions used in this note. The concept of state, based on an extension of Nerode equivalence to general dynamical systems, is given in Section~\ref{sec:state}. The main result of this note\newstuff{, namely, the existence of a state-space representation for causal, time invariant systems,} is provided in Section~\ref{sec:state-space-representation}\newstuff{, the minimality of the obtained representation is treated in Section~\ref{sec:minimality}}, and Section~\ref{sec:conclusions} concludes the paper.
  
%
%

\section{Notation and definitions} \label{sec:notation}

In the sequel, $\mathbb{Z}$ is the set of integers $\{\dots, -2, -1, 0, 1, 2, \dots \}$, $\mathbb{N}$ is the set of natural numbers $\{1, 2, \dots \}$, $\mathbb{N}_0 := \mathbb{N} \cup \{ 0 \} = \{ 0, 1, 2, \dots \}$, and $\mathbb{Z}_- := \{ 0, -1, -2, \dots \}$.
If $A, B$ are sets, $B^A$ denotes the set of all functions from $A$ to $B$; if $A$ is a subset of $\mathbb{Z}$, $B^A$ can be seen as a set of vectors or ordered tuples (if $A$ is finite) or sequences (\eg, $B^{\mathbb{N}}$ is the set of sequences $b = (b_1, b_2, b_3, \dots)$, with $b_k \in B$ for all $k \in \mathbb{N}$).
\newstuff{For a set $A \subseteq \mathbb{Z}$, and $n \in \mathbb{Z}$, let $A + n := \{ k + n\colon k \in A \}$.}
Given two mappings $f\colon X \to Y$, $g\colon Y \to Z$, their composition is denoted as $g \circ f\colon X \to Z$, or simply as $g f$, when it is clear from the context.

\begin{defn}
Given sets $U$, $Y$, a \emph{(discrete-time, non-autonomous) dynamical system from $U$ to $Y$} is a mapping $T\colon U^\mathbb{Z} \to Y^\mathbb{Z}$.
\end{defn}

In words, a dynamical system is simply a function that maps sequences in $U$ into sequences in $Y$.

\smallskip
Given a set $X$, let $q_X\colon X^\mathbb{Z} \to X^\mathbb{Z}$ be the mapping that satisfies $(q_X x)(n) = x(n+1)$ for every $x \in X^\mathbb{Z}$, $n \in \mathbb{N}$; $q_X$ is called the \emph{shift operator in $X^\mathbb{Z}$}.

Note that $q_X$ is bijective, with inverse $q_X^{-1}$. By convention, we denote the $n$-fold composition $q_X \circ \cdots \circ q_X$ as $q_X^n$ ($n \in \mathbb{N}$), and similarly $q_X^{-n}$ is the $n$-fold composition $q_X^{-1} \circ \cdots \circ q_X^{-1}$.

\begin{defn}
Given sets $U,Y$, a dynamical system $T$ from $U$ to $Y$ is \emph{time invariant} if $q_Y^n \circ T = T \circ q_U^n$ for all $n \in \mathbb{Z}$.
\end{defn}

Thus, for a time invariant system, if the input sequence $u\in U^\mathbb{Z}$ is shifted by $n$ time units, the output $T u$ is also shifted by $n$ time units.

\newstuff{
\begin{remark}
It is possible to define a time invariant system as mapping sequences in $\mathbb{N}_0$ rather than in $\mathbb{Z}$, by including ``initial conditions'' as part of its definition. However, doing so implies somehow already having a state-space realization of the system (whose existence we aim to establish in this paper), since those initial conditions provide the state of the system at time 0 (albeit not in a minimal realization) according to the definition of state to be given in the next section.
\end{remark}}

\smallskip
Consider a non-empty set $U$, and $A \subseteq \mathbb{Z}$. \newstuff{Given a fixed element $o \in U$,} let us define the \emph{projection operator onto $U^A$}, $\pi_A\colon U^{\mathbb{Z}} \to U^{\mathbb{Z}}$, where\footnote{The projection operator $\pi_A$ depends of course on $U$ and $o \in U$, but these dependencies will be omitted to keep the notation simple.} (for every $x \in U^{\mathbb{Z}}$, $n \in \mathbb{Z}$)
\begin{align*}
(\pi_A x)(n) := \begin{cases}
x(n), & \text{ if } n \in A \\
o,    & \text{ otherwise}.
\end{cases}
\end{align*}

Note that, for every $k,n \in \mathbb{Z}$,
\begin{align*}
(\pi_A q_U^n u) (k)
&= \begin{cases}
(q_U^n u)(k), & \text{ if } k \in A \\
o, & \text{otherwise}
\end{cases} \\
&= \begin{cases}
u(k+n), & \text{ if } k \in A \\
o, & \text{otherwise}
\end{cases} \\
&= q_U^n \begin{cases}
u(k), & \text{ if } k - n \in A \\
o, & \text{otherwise}
\end{cases} \\
&= (q_U^n \pi_{A + n} u) (k),
\end{align*}
thus $\pi_A q_U^n = q_U^n \pi_{A + n}$. In particular,
\begin{align} \label{eq:q_pi}
\begin{aligned}
\pi_{\mathbb{Z}_-} q_U^n &= q_U^n \pi_{\{\dots, n-1, n \}} \\
\pi_{\mathbb{N}_0} q_U^n &= q_U^n \pi_{\{n, n+1, \dots\}}, \qquad \text{for all } n \in\mathbb{Z}.
\end{aligned}
\end{align}

\begin{defn}
A time invariant dynamical system $T$ on a set $U$ is \emph{causal} if $(T \pi_{\mathbb{Z}_-} u)(0) = (T u)(0)$ for all $u \in U^{\mathbb{Z}}$.
\end{defn}

In words, the output of a causal system at time $0$ only depends on the values of its input up to time $0$.

\smallskip
Notice that, due to the time invariance of $T$ and \eqref{eq:q_pi}, its causality implies that, for all $u \in U^{\mathbb{Z}}$ and $n \in \mathbb{Z}$,
\begin{align*}
(T \pi_{\{\dots, n-2,n-1,n\}} u)(n)
&= (q_Y^n T \pi_{\{\dots, n-2,n-1,n\}} u)(0) \\
&= (T q_U^n \pi_{\{\dots, n-2,n-1,n\}} u)(0) \\
&= (T \pi_{\mathbb{Z}_-} q_U^n u)(0) \\
&= (T q_U^n u)(0) \\
&= (q_Y^n T u)(0) \\
&= (T u)(n).
\end{align*}
If $u_1,u_2 \in U^{\mathbb{Z}}$, let us define their \emph{concatenation at time $n$}, ${\concat{u_1}{u_2}{n}}$, as an element of $U^{\mathbb{Z}}$ such that, for every $k \in \mathbb{Z}$,
\begin{align*}
(\concat{u_1}{u_2}{n})(k) := \begin{cases}
u_1(k), & \text{ if } k < n \\
u_2(k), & \text{ if } k \geq n.
\end{cases}
\end{align*}

Note that, for every $i,k,n \in \mathbb{Z}$ and $u_1,u_2 \in U^{\mathbb{Z}}$,
\begin{align*}
[q_U^n(\concat{u_1}{u_2}{k})](i)
&= (\concat{u_1}{u_2}{k})(i+n) \\
&= \begin{cases}
u_1(i+n), & \text{ if } i+n < k \\
u_2(i+n), & \text{ if } i+n \geq k
\end{cases} \\
&= \begin{cases}
(q_U^n u_1)(i), & \text{ if } i < k - n \\
(q_U^n u_2)(i), & \text{ if } i \geq k - n
\end{cases} \\
&= (\concat{q_U^n u_1}{q_U^n u_2}{k-n})(i),
\end{align*}
hence
\begin{align} \label{eq:q_plus}
q_U^n(\concat{u_1}{u_2}{k}) = \concat{q_U^n u_1}{q_U^n u_2}{k-n}.
\end{align}

Finally, we need to define the \emph{insertion of a value $a \in U$ in a sequence $u \in U^{\mathbb{Z}}$ at time $n \in \mathbb{Z}$}, $(u, a)_n$, as an element in $U^{\mathbb{Z}}$ such that
\begin{align*}
(u,a)_n(k) := \begin{cases}
u(k), & \text{ if } k \neq n \\
a,    & \text{ if } k = n.
\end{cases}
\end{align*}
For all $u_1,u_2 \in U^{\mathbb{Z}}$, $a \in U$ and $i \in \mathbb{Z}$, it holds that
\begin{align*}
[\concat{u_1}{(u_2,a)_0}{0}](i)
&= \begin{cases}
u_1(i), & i < 0 \\
(u_2,a)_0(i) & i \geq 0
\end{cases} \\
&= \begin{cases}
u_1(i), & i < 0 \\
a, & i = 0 \\
u_2(i) & i > 0
\end{cases} \\
&= \begin{cases}
(u_1, a)_0(i), & i < 1 \\
u_2(i) & i \geq 1
\end{cases} \\
&= [\concat{(u_1, a)_0}{u_2}{1}](i),
\end{align*}
thus
\begin{align} \label{eq:concat_insert}
\concat{u_1}{(u_2,a)_0}{0} = \concat{(u_1, a)_0}{u_2}{1}.
\end{align}

\section{Concept of state} \label{sec:state}

Consider a causal, time invariant dynamical system $T$ from $U$ to $Y$.
Our goal is to define the notion of the ``state'' of $T$ at a given time $n$.
This can be achieved by extending the concept of Nerode equivalence, developed in \cite{Nerode-58} for linear systems, as done in the following definition.

\begin{defn}
Let $u_1,u_2 \in U^{\mathbb{Z}}$. The sequences $u_1$ and $u_2$ are \emph{Nerode equivalent (under $T$) at time $0$} if, for every $z \in U^{\mathbb{Z}}$, it holds that $\pi_{\mathbb{N}_0} T ({\concat{u_1}{z}{0}}) = \pi_{\mathbb{N}_0} T (\concat{u_2}{z}{0})$; if this condition holds, we use the notation $u_1 \nerode{0} u_2$.
\end{defn}

It can be seen that Nerode equivalence at time $0$ is an equivalence relation, \ie, it is \emph{reflexive} ($u \nerode{0} u$ for all $u \in U^{\mathbb{Z}}$), \emph{symmetric} ($u_1 \nerode{0} u_2$ implies $u_2 \nerode{0} u_1$ for all $u_1,u_2 \in U^{\mathbb{Z}}$) and \emph{transitive} ($u_1 \nerode{0} u_2$ and $u_2 \nerode{0} u_3$ imply $u_1 \nerode{0} u_3$ for all $u_1,u_2,u_3 \in U^{\mathbb{Z}}$).
The equivalence class of $u \in U^{\mathbb{Z}}$ is denoted $[u]_0 := \{ \tilde{u} \in U^{\mathbb{Z}}\colon u \nerode{0} \tilde{u} \}$. This definition can be extended to an arbitrary time $n$:

\begin{defn}
$u_1,u_2 \in U^{\mathbb{Z}}$ are said to be \emph{Nerode equivalent (under $T$) at time $n$} ($n \in \mathbb{Z}$), denoted $u_1 \nerode{n} u_2$, if $q_U^{-n} u_1 \nerode{0} q_U^{-n} u_2$.
The corresponding equivalence class of $u$ is denoted $[u]_n$, and it satisfies $[u]_n = q_U^n [q_U^{-n} u]_0$.
\end{defn}

The proof of the last statement is the following: $[u]_n = \{ \tilde{u} \in U^{\mathbb{Z}}\colon q_U^{-n} u \nerode{0} q_U^{-n} \tilde{u} \} = q_U^n \{ q_U^{-n} \tilde{u} \in U^{\mathbb{Z}}\colon q_U^{-n} u \nerode{0} q_U^{-n} \tilde{u} \} = q_U^n \{ u' \in U^{\mathbb{Z}}\colon q_U^{-n} u \nerode{0} u' \} = q_U^n [q_U^{-n} u]_0$.

\medskip
Roughly speaking, if $u_1,u_2$ are two Nerode equivalent inputs at time $0$,  the output of $T$ for non-negative times will be the same for $u_1$ and $u_2$ as long as $u_1(n) = u_2(n)$ for all $n \geq 0$. In this sense, we say that for both $u_1$ and $u_2$ the system $T$ is ``in the same state'' at time $0$.

\section{Existence of a state-space representation} \label{sec:state-space-representation}

In this section we will establish that every causal, time-invariant system has a state-space realization. Let us first define the notion of a general state-space realization.

\begin{defn}
Given sets $U, Y$, and a causal, time invariant system $T\colon U^{\mathbb{Z}} \to Y^{\mathbb{Z}}$, a \emph{state-space realization of $T$, with state space $X$}, is a pair $(f, g)$, with $f\colon U \times X \to X$ and $g\colon U \times X \to Y$, such that for every $u \in U^{\mathbb{Z}}$, there is an $x \in X^{\mathbb{Z}}$ for which
\begin{align} \label{eq:state_eqn}
\begin{aligned}
x(n+1) &= f(u(n), x(n)) \\
y(n) &= g(u(n), x(n)), \quad n \in \mathbb{Z},
\end{aligned}
\end{align}
where $y = T u \in Y^{\mathbb{Z}}$.
\end{defn}

Throughout this section, we will consider a fixed causal, time invariant dynamical system $T$. In Section~\ref{sec:state} we have seen that the state of $T$ at time $0$ can be represented by an equivalence class $[u]_0$. This means that the ``state space'' of $T$ can be represented by the quotient set
\begin{align*}
U^{\mathbb{Z}} / \nerode{0}\; := \{ [u]_0\colon u \in U^{\mathbb{Z}} \}.
\end{align*}
Notice that since $U$ and $Y$ are arbitrary sets, without, \eg, a pre-specified vector space or manifold structure, the state space $U^{\mathbb{Z}} / \nerode{0}$ is at this point merely an abstract set.

To arrive at a state-space representation of $T$, let us define the \emph{canonical projection of $U^{\mathbb{Z}}$ onto $U^{\mathbb{Z}} / \nerode{0}$}, $P\colon U^{\mathbb{Z}} \to U^{\mathbb{Z}} / \nerode{0}$, as $P u := [u]_0$ for all $u \in U^{\mathbb{Z}}$.

Next, we provide two lemmas stating the existence of functions which will play the role of a one-step state transition function and of an output mapping, respectively:

\begin{lemma} \label{lem:existence_f}
There exists a function $\bar{f}\colon U \times U^{\mathbb{Z}} / \nerode{0}\; \to U^{\mathbb{Z}} / \nerode{0}$ such that $\bar{f}(a, [u]_0) := [q_U (u,a)_0]_0$ for all $a \in U$, $u \in U^{\mathbb{Z}}$.
\end{lemma}

\begin{proof}
We need to show that a function $\bar{f}$ satisfying $\bar{f}(a, [u]_0) = [q_U (u,a)_0]_0$ is well-defined. To this end, let $u_1, u_2 \in U^{\mathbb{Z}}$ such that $u_1 \nerode{0} u_2$. Then, taking some $a \in U$ and $z \in U^{\mathbb{Z}}$,
\begin{align*}
&u_1 \nerode{0} u_2 \\
&\Rightarrow \quad
\pi_{\mathbb{N}_0} T[\concat{u_1}{(q_U^{-1} z, a)_0}{0}] = \pi_{\mathbb{N}_0} T[\concat{u_2}{(q_U^{-1} z, a)_0}{0}] \\
&\Rightarrow \quad
\pi_{\mathbb{N}} T[\concat{u_1}{(q_U^{-1} z, a)_0}{0}] = \pi_{\mathbb{N}} T[\concat{u_2}{(q_U^{-1} z, a)_0}{0}] \\
&\Rightarrow \quad
\pi_{\mathbb{N}} T[\concat{(u_1, a)_0}{q_U^{-1} z}{1}] = \pi_{\mathbb{N}} T[\concat{(u_2, a)_0}{q_U^{-1} z}{1}] \\
&\Rightarrow \quad
q_U \pi_{\mathbb{N}} T[\concat{(u_1, a)_0}{q_U^{-1} z}{1}] = q_U \pi_{\mathbb{N}} T[\concat{(u_2, a)_0}{q_U^{-1} z}{1}] \\
&\Rightarrow \quad
\pi_{\mathbb{N}_0} q_U T[\concat{(u_1, a)_0}{q_U^{-1} z}{1}] = \pi_{\mathbb{N}_0} q_U T[\concat{(u_2, a)_0}{q_U^{-1} z}{1}] \\
&\Rightarrow \quad
\pi_{\mathbb{N}_0} T[\concat{q_U (u_1, a)_0}{q_U q_U^{-1} z}{0}] = \pi_{\mathbb{N}_0} T[\concat{q_U (u_2, a)_0}{q_U q_U^{-1} z}{0}] \\
&\Rightarrow \quad
\pi_{\mathbb{N}_0} T[\concat{q_U (u_1, a)_0}{z}{0}] = \pi_{\mathbb{N}_0} T[\concat{q_U (u_2, a)_0}{z}{0}],
\end{align*}
where we have used \eqref{eq:q_pi}, \eqref{eq:q_plus} and \eqref{eq:concat_insert}. Since $z$ was arbitrary, this shows that $q_U(u_1, a)_0 \nerode{0} q_U(u_2, a)_0$, \ie, $[q_U(u_1, a)_0]_0 = [q_U(u_2, a)_0]_0$. Therefore, $\bar{f}(a, [u_1]_0) = \bar{f}(a, [u_2]_0)$, so $f$ is well-defined.
\end{proof}

\begin{lemma} \label{lem:existence_g}
There exists a function $\bar{g}\colon U \times U^{\mathbb{Z}} / \nerode{0}\, \to Y$ such that
\begin{align*}
(T u)(0) = \bar{g}(u(0), [u]_0), \quad u \in U^{\mathbb{Z}}.
\end{align*}
\end{lemma}

\begin{proof}
One can define $\bar{g}$ as $\bar{g}(a, [u]_0) := (T (u,a)_0)(0)$ for all $a \in U$, $u \in U^{\mathbb{Z}}$; this satisfies the statement of the theorem, since $[u]_0 = [(u,a)_0]_0$. This function is well-defined, since if $u_1 \nerode{0} u_2$, and one picks any $z \in U^{\mathbb{Z}}$ such that $z(0) = a$,
\begin{align*}
&\pi_{\mathbb{N}_0} T (\concat{u_1}{z}{0}) = \pi_{\mathbb{N}_0} T (\concat{u_2}{z}{0}) \\
&\Rightarrow \quad
[T (\concat{u_1}{z}{0})](0) = [T (\concat{u_2}{z}{0})](0) \\
&\Rightarrow \quad
[T \pi_{\mathbb{Z}_-} (\concat{u_1}{z}{0})](0) = [T \pi_{\mathbb{Z}_-} (\concat{u_2}{z}{0})](0) \;\; \text{(due to causality)} \\
&\Rightarrow \quad
[T \pi_{\mathbb{Z}_-} (u_1,a)_0](0) = [T \pi_{\mathbb{Z}_-} (u_2,a)_0](0) \\
&\Rightarrow \quad
(T (u_1,a)_0)(0) = (T (u_2,a)_0)(0) \quad \text{(due to causality)}.
\end{align*}

\vspace{-6mm}
\end{proof}

\medskip
Lemmas~\ref{lem:existence_f} and \ref{lem:existence_g} lead to the main result of this note:

\smallskip
\begin{theorem} \label{thm:realization}
Consider a causal, time invariant dynamical system $T$ from $U$ to $Y$. Then, there exist functions $\bar{f}\colon U \times U^{\mathbb{Z}} / \nerode{0}\; \to U^{\mathbb{Z}} / \nerode{0}$ and $\bar{g}\colon U \times U^{\mathbb{Z}} / \nerode{0}\; \to Y^{\mathbb{Z}}$ such that for all $u \in U^{\mathbb{Z}}$ and $y = T u\in Y^{\mathbb{Z}}$, we have that
\begin{align} \label{eq:state_space}
\begin{aligned}
x(n+1) &= \bar{f}(u(n), x(n)) \\
y(n) &= \bar{g}(u(n), x(n)), \qquad n \in \mathbb{Z},
\end{aligned}
\end{align}
where $x \in (U^{\mathbb{Z}} / \nerode{0})^{\mathbb{Z}}$.
\end{theorem}

\begin{proof}
By Lemmas~\ref{lem:existence_f} and \ref{lem:existence_g}, the expression $y = T u$, where $u \in U^{\mathbb{Z}}$ and $y \in Y^{\mathbb{Z}}$, can be described in the form \eqref{eq:state_space}, where $x(n) = [u]_n$ is an entry of the sequence $x \in (U^{\mathbb{Z}} / \nerode{0})^{\mathbb{Z}}$.
\end{proof}

\smallskip
Equations~\eqref{eq:state_space} correspond to what we call the \emph{Nerode (state-space) realization} of $T$.

\newstuff{
\begin{remark}
The formulation and derivation of the state-space realization in Theorem~\ref{thm:realization} resembles that of canonical recognizers of regular languages; see, \eg, \cite[Section~2.3]{Wonham-Cai-19}.
\end{remark}

\begin{remark}
Notice that, while the representation given by Theorem~\ref{thm:realization} appears to be quite abstract, under certain conditions it can be given a fairly concrete form. For example, if $U,Y$ are finite dimensional vector spaces, and $T$ is a linear mapping, $U^{\mathbb{Z}} / \nerode{0}$ inherits a linear structure; this is studied, \eg, in \cite[Chapter~5]{Kailath-80}, where the Nerode equivalence is used to construct explicit minimal state-space realizations for finite-dimensional linear systems.
\end{remark}}

\section{Minimality of the Nerode realization} \label{sec:minimality}

In this section we establish the minimality of the Nerode realization given by Theorem~\ref{thm:realization}. Consider a state-space realization $(f,g)$ of a system $T\colon U^{\mathbb{Z}} \to Y^{\mathbb{Z}}$, with state space $X$. Define the \emph{controllable subset}
\begin{align*}
X_c := &\{a \in X\colon \text{there is a } u \in U^{\mathbb{Z}} \text{ and an } x \in X^{\mathbb{Z}} \text{ with } x(0) = a \\
&\qquad\qquad \text{ such that the first eqn. of \eqref{eq:state_eqn} holds for all } n \in \mathbb{Z} \}.
\end{align*}
In words, $X_c$ corresponds to the set of values that the state $x$ can take at time $0$. Notice that the state space $X$ can be restricted to $X_c$ without loss of generality, since if for some $u \in U^{\mathbb{Z}}$, a state $x$ satisfying \eqref{eq:state_eqn} is such that $x(n) \notin X_c$ for some $n \in \mathbb{Z}$, then the input $q_U^{-n} u$ would yield a state $q_X^{-n} x$, where $(q_X^{-n} x)(0) = x(n) \notin X_c$, contradicting the definition of $X_c$.

The following result establishes that the realization given by Theorem~\ref{thm:realization} is minimal in a fairly general sense, as it establishes a surjective mapping between any realization $(f,g)$ and the Nerode realization.

\begin{theorem} \label{thm:minimality}
Consider a causal, time invariant dynamical system $T$ from $U$ to $Y$, with a state-space realization $(f,g)$ in the state space $X$. Then, there exists a surjective mapping $\bar{P}\colon X_c \to U^{\mathbb{Z}} / \nerode{0}$ such that $\bar{f}(b, \bar{P}(a)) = \bar{P}(f(b, a))$ and $\bar{g}(b, \bar{P}(a)) = g(b, a)$ for all $b \in U$ and $a \in X_c$; see Figure~\ref{fig:PN}.
\end{theorem}

\begin{figure}[t]
\centering
\begin{tikzcd}
X_c \arrow[d, "\bar{P}"', dashed, two heads] \arrow[rr, "{f(b, \cdot)}"]      &  & X_c \arrow[d, "\bar{P}", dashed, two heads] \arrow[rr, "{g(b,\cdot)}"]         &  & Y \\
U^{\mathbb{Z}} / \stackrel{0}{\sim} \arrow[rr, "{\bar{f}(b, \cdot)}"] &  & U^{\mathbb{Z}} / \stackrel{0}{\sim} \arrow[rru, "{\bar{g}(b,\cdot)}"'] &  &  
\end{tikzcd}

\caption{Commutative diagram for $\bar{P}$ (cf. Theorem~\ref{thm:minimality}) \cite{MacLane-Birkhoff-88}. The double- headed arrows represent surjective maps.}
\label{fig:PN}
\end{figure}

\begin{proof}
Given an $a \in X_c$, let $u \in U^{\mathbb{Z}}$ and $x \in X^{\mathbb{Z}}$ as given in the definition of $X_c$, and define $\bar{P}(a) := [u]_0$. This mapping is clearly surjective, and to show that it is well-defined, let $u_1, u_2 \in U^{\mathbb{Z}}$ and $x_1,x_2 \in X^{\mathbb{Z}}$ be such that $(u_1, x_1)$ and $(u_2, x_2)$ satisfy \eqref{eq:state_eqn}, and $x_1(0) = x_2(0) = a$. Then, for every $z \in U^{\mathbb{Z}}$, the inputs $\concat{u_1}{z}{0}$ and $\concat{u_2}{z}{0}$ generate states $\tilde{x}_1, \tilde{x}_2 \in X^{\mathbb{Z}}$ and outputs $\tilde{y}_1, \tilde{y}_2 \in Y^{\mathbb{Z}}$, according to \eqref{eq:state_eqn}, such that $\pi_{\mathbb{N}_0} \tilde{x}_1 = \pi_{\mathbb{N}_0} \tilde{x}_2$ and $\pi_{\mathbb{N}_0} \tilde{y}_1 = \pi_{\mathbb{N}_0} \tilde{y}_2$, thus $\pi_{\mathbb{N}_0} T ({\concat{u_1}{z}{0}}) = \pi_{\mathbb{N}_0} T (\concat{u_2}{z}{0})$, or $[u_1]_0 = [u_2]_0$.

To establish the properties of $\bar{P}$, fix an $a \in X_c$ and a point $b \in U$, and let $u \in U^{\mathbb{Z}}$ and $x \in X^{\mathbb{Z}}$ as given in the definition of $X_c$. Let $\tilde{u} = (u,b)_0$, and define $\tilde{x} \in X^{\mathbb{Z}}$ such that $\tilde{x}(n) = x(n)$ for all $n \leq 0$, and $\tilde{x}(n+1) = f(\tilde{u}(n), \tilde{x}(n))$ for all $n > 0$. Then,
\begin{align*}
\bar{f}(b, \bar{P}(a))
&= \bar{f}(b, [u]_0) \\
&= [q_U (u,b)_0]_0 \\
&= [q_U \tilde{u}]_0 \\
&= \bar{P}(\tilde{x}(1)) \\
&= \bar{P}(f(b,a)).
\end{align*}
Similarly,
\begin{align*}
\bar{g}(b, \bar{P}(a))
= \bar{g}(b, [u]_0)
= T(u,b)_0
= g(b,a).
\end{align*}
This concludes the proof of the theorem.
\end{proof}

Theorem~\ref{thm:minimality} provides the most general statement on the minimality of the Nerode realization without relying on additional structure of $U,Y$ or $(f,g)$.
It implies, for instance, that if the cardinality of the state space given by Theorem~\ref{thm:realization}, $U^{\mathbb{Z}} / \nerode{0}$, is finite, it has the smallest number of elements among all realizations of $T$.
On the other hand, if $U,Y$ are finite-dimensional vector spaces, and $(f,g)$ is a realization of $T$ such that $f,g$ are linear functions, it can be shown (see, \eg, \cite[Chapter~5]{Kailath-80}) that the realization of Theorem~\ref{thm:realization} can be given a linear structure, of dimension not larger than that of the state-space realization described by $(f,g)$. For more general nonlinear setups, the following corollary provides a weaker result in the case of dynamical systems on topological manifolds, showing that the sense of minimality established in Theorem~\ref{thm:minimality} coincides with the notion of minimal dimension for points in $X_c$ which are \emph{locally observable}\footnote{Local observability is not a necessary assumption, but removing it may require machinery that is beyond the scope of this paper; in this regard, Corollary~\ref{cor:minimal_smooth} is only meant as a example of the type of results one could obtain for continuous/smooth dynamical systems.}.

\begin{defn}
Let $(f,g)$ be a state-space realization of a system $T\colon U^{\mathbb{Z}} \to Y^{\mathbb{Z}}$, with state space $X$ endowed with a topology. A point $a \in X$ is \emph{locally observable}\footnote{This definition of local observability is slightly different from the standard one, \cf \cite{albertini-d'alessandro-96}.} if there is a neighborhood $V$ of $a$ such that for all $a_1, a_2 \in V$, $a_1 \neq a_2$, there is a sequence $u \in U^{\mathbb{N}_0}$ such that at least one of the following identities does \emph{not} hold:
\begin{align*}
g(u(0), a_1) &= g(u(0), a_2) \\ 
g(u(1), f(u(0), a_1)) &= g(u(1), f(u(0), a_2)) \\
g(u(2), f(u(1),f(u(0), a_1))) &= g(u(2), f(u(1),f(u(0), a_2))) \\
&\;\,\vdots
\end{align*}
\end{defn}

\smallskip
Notice that if $a \in X$ is locally observable, so is $f(b,a) \in X$ for every $b \in U$. 

\begin{corollary} \label{cor:minimal_smooth}
Under the assumptions of Theorem~\ref{thm:minimality}, assume that $U,X_c, Y$ are finite-dimensional topological manifolds, and that $f,g$ are continuous mappings (when restricted to $X_c$). Then, $U^{\mathbb{Z}}/ \nerode{0}$ can be given a differentiable structure at every point $\bar{P}(a)$ where $a \in X_c$ is locally observable, of the same dimension as that of $X_c$, such that $\bar{P}, \bar{f}, \bar{g}$ are continuous mappings at $a$.
\end{corollary}

\begin{proof}
Let $a \in X_c$ be locally observable, and $V$ its associated neighborhood in $X_c$.
The idea of the proof is to endow $U^{\mathbb{Z}}/\nerode{0}$ with a manifold structure inherited from $X_c$ through the mapping $\bar{P}$ restricted to $V$. This mapping can be factorized as \cite[Section~I.9]{MacLane-Birkhoff-88}

\begin{center}
\begin{tikzcd}
V \arrow[rr, "P_{\ker(\bar{P})}"] \arrow[rrd, "\bar{P}"', two heads] & & V / \ker(\bar{P}) \arrow[d, "\tilde{P}", dashed, hook] \\
                                                               & & U^{\mathbb{Z}} / \stackrel{0}{\sim}             
\end{tikzcd}
\end{center}

\noindent where $\ker(\bar{P})$ is an equivalence relation in $V$ (called the \emph{equivalence kernel} of $\bar{P}$), such that for all $a_1,a_2 \in V$, $(a_1, a_2) \in \ker(\bar{P})$ iff $\bar{P}(a_1) = \bar{P}(a_2)$; $P_{\ker(\bar{P})}$ is the canonical projection of $X_c$ onto $X_c / \ker(\bar{P})$ (\ie, for every $a \in X_c$, $P_{\ker(\bar{P})}(a)$ is the equivalence class of $a$ in $X_c / \ker(\bar{P})$); and $\tilde{P}$ is a bijective mapping. Thus, we need to define a manifold structure on $V / \ker(\bar{P})$ such that $P_{\ker(\bar{P})}$ is smooth. To this end, note that $\ker({\bar{P}})$ is also the equivalence kernel of a mapping $F\colon U \to C(U^{\mathbb{N}_0}, Y^{\mathbb{N}_0})$ given by $(F(a))(u) = (g(u(0), a), g(u(1), f(u(0), a)), \dots)$; due to the assumption of local observability, $F$ is locally injective, so $V / \ker(\bar{P})$ is equal to $V$, and $P_{\ker(\bar{P})}$ is the identity mapping. Therefore, $U^{\mathbb{Z}} / \stackrel{0}{\sim}$ inherits around $\bar P(a)$, through $\bar P$, the same manifold structure of $X_c$ around $a$, for which the dimension of $U^{\mathbb{Z}} / \stackrel{0}{\sim}$ at $\bar P(a)$ is the same as of $X_c$ at $a$. Furthermore, $\bar{f}$ and $\bar{g}$ are continuous at $a$, since $f(b,a)$ is locally observable for every $b$ (so the same arguments apply to $\bar P(f(b,a))$).
%
%
%
%
\end{proof}

\smallskip
Under the assumptions of Corollary~\ref{cor:minimal_smooth}, in case $X$ is a topological manifold, and $X_c$ is a submanifold of it, it follows that the dimension of $U^{\mathbb{Z}} / \stackrel{0}{\sim}$ at $\bar P(a)$ is not larger than that of $a$ in $X$. Also, notice that it is possible to strengthen the corollary by imposing additional assumptions that ensure that $X_c$ can be further reduced to a locally observable submanifold; see, \eg, \cite{albertini-d'alessandro-96}.

\section{Conclusions} \label{sec:conclusions}

We have established that every non-autonomous discrete-time dynamical system which satisfies the properties of time invariance and causality admits a state-space representation. The result is fairly general, as it does not require the system to have any additional structure (\eg, continuity or differentiability). \newstuff{We have also discussed the minimality of the obtained state-space realization.}

The condition on time invariance can potentially be relaxed, at the expense of allowing functions $f$ and $g$ of the state-space representation to depend on time. However, the requirement of causality is crucial, since every discrete-time system which has a state-space representation is necessarily causal.

\bibliographystyle{plain} 
\bibliography{refs}
\end{document}